\documentclass[10pt,onecolumn,oneside]{article}
\usepackage{lmodern}
\usepackage{amsmath,amssymb,amsfonts,amsthm,dsfont}
\usepackage{authblk}
\usepackage{stackengine}
\usepackage{verbatim}
\usepackage{longtable}
\usepackage{algorithm,algorithmicx,algpseudocode,float}
\usepackage{fancyvrb}
\usepackage{hyperref}
\usepackage{graphicx}
\usepackage{xcolor}
\usepackage{multicol}
\usepackage{adjustbox}
\usepackage{array}

\newcolumntype{R}[2]{%
    >{\adjustbox{angle=#1,lap=\width-(#2)}\bgroup}%
    l%
    <{\egroup}%
}

\allowdisplaybreaks 

\theoremstyle{plain} 
\theoremstyle{plain} \newtheorem{proposition}{\textbf{Proposition}}
\theoremstyle{remark} 
\theoremstyle{plain} \newtheorem{theorem}{\textbf{Theorem}}
\theoremstyle{plain} 
\theoremstyle{plain} 
\theoremstyle{definition} \newtheorem{definition}{\textbf{Definition}}
\theoremstyle{plain} \newtheorem{conjecture}{\textbf{Conjecture}}

\newcommand{\rmv}[1]{}

\makeatletter
\let\@@pmod\pmod
\DeclareRobustCommand{\pmod}{\@ifstar\@pmods\@@pmod}
\def\@pmods#1{\mkern4mu({\operator@font mod}\mkern 6mu#1)}
\makeatother

\makeatletter
\newenvironment{breakablealgorithm}
  {
   \begin{center}
     \refstepcounter{algorithm}
     \hrule height.8pt depth0pt \kern2pt
     \renewcommand{\caption}[2][\relax]{
       {\raggedright\textbf{\ALG@name~\thealgorithm} ##2\par}%
       \ifx\relax##1\relax 
         \addcontentsline{loa}{algorithm}{\protect\numberline{\thealgorithm}##2}%
       \else 
         \addcontentsline{loa}{algorithm}{\protect\numberline{\thealgorithm}##1}%
       \fi
       \kern2pt\hrule\kern2pt
     }
  }{
     \kern2pt\hrule\relax
   \end{center}
  }
\makeatother

\floatname{algorithm}{Procedure}

\renewcommand{\algorithmicreturn}[1]{\bgroup\\  ~#1\egroup}
\renewcommand{\algorithmiccomment}[1]{\bgroup\hfill//~#1\egroup}
\title{Finding linearly generated subsequences}
\author{Claude Gravel}
\affil{\stackunder{\small{National Institute of Informatics}}{\stackunder{\small{Tokyo, Japan}}{\stackunder{\small{\texttt{claudegravel1980@gmail.com}}}{\small{Currently working at EAGLYS (Tokyo, Japan)}}}}}
\author{Daniel Panario}
\affil{\stackunder{\small{School of Mathematics and Statistics}}{\stackunder{\small{Carleton University, Canada}}{\mbox{\small{\texttt{daniel@math.carleton.ca}}}}}}
\author{Bastien Rigault}
\affil{\stackunder{\small{National Institute of Informatics}}{\stackunder{\small{Tokyo, Japan}}{\mbox{\small{\texttt{rgaultb@gmail.com}}}}}}
\date{\today}

\begin{document}

\maketitle

\begin{abstract}
We develop a new algorithm to compute determinants of all possible Hankel matrices made up from a given finite length sequence over a finite field. Our algorithm fits within the dynamic programming paradigm by exploiting new recursive relations on the determinants of Hankel matrices together with new observations concerning the distribution of zero determinants among the possible matrix sizes allowed by the length of the original sequence. The algorithm can be used to isolate \emph{very} efficiently linear shift feedback registers hidden in strings with random prefix and random postfix for instance and, therefore, recovering the shortest generating vector. Our new mathematical identities can be used also in any other situations involving determinants of Hankel matrices. We also implement a parallel version of our algorithm. We compare our results empirically with the trivial algorithm which consists of computing determinants for each possible Hankel matrices made up from a given finite length sequence. Our new accelerated approach on a single processor is faster than the trivial algorithm on 160 processors for input sequences of length 16384 for instance.

\textbf{Keywords: } generating polynomial, linear algebra over finite fields, quotient-difference tables, Hankel matrices, linear shift feedback registers, pattern substrings, Berlekamp-Massey algorithm
\end{abstract}

\section{Notation, facts and definitions}\label{sect:intro}

Let $q$ be a prime power, $n>0$, and ${\bf x}=(x_i)_{i=0}^{n-1}\in\mathbb{F}^{n}_{q}$. For integers $1\leq j\leq\big\lceil{\frac{n}{2}}\big\rceil$ and $j-1\leq i< n-j+1$, define the matrix $\mathbf{X}_{i,j}$ by
\begin{align}
\mathbf{X}_{i,j}&=
\left(\begin{array}{cccc}
x_{i}     & \ldots  & x_{i+j-2} & x_{i+j-1}\\
x_{i-1}   & \ldots  & x_{i+j-3} & x_{i+j-2}\\
\vdots    & \ddots  & \vdots    & \vdots\\
x_{i-j+1} & \ldots  & x_{i-1}   & x_{i}
\end{array}\right).\label{matrepxij}
\end{align}
By convention, we let $\mathbf{X}_{i,0}=1$ for $0\leq i<n$. Every matrix
$\mathbf{X}_{i,j}$ is a Hankel matrix.

Hankel matrices have a large number of applications in applied mathematics. In this paper we are interested in Hankel matrices over finite fields. We explore the well known connection of Hankel matrices and sequences over finite fields; for an introductory explanation see Section 8.6 in \cite{lidl_niederreiter_1996}. Hankel matrices are strongly connected to coprime polynomials over finite fields. Indeed, the probability of two monic polynomials of positive degree $n$ over $\mathbb{F}_{q}$ to be relatively prime is the same as the uniform probability for an $n \times n$ Hankel matrix over $\mathbb{F}_{q}$ be nonsingular \cite{GaoPanario2004,Armas2011}. Elkies \cite{Elkies2002} studies the probability of Hankel matrices over finite fields be nonsingular when independent biased entries are used for the matrix. An algorithm to generate a class of Hankel matrices called superregular (that are related to MDS codes) is given in \cite{Raemy2015}. Finally, we point out that several results and  applications of Hankel matrices over finite fields are given in the Handbook of Finite Fields \cite{MullenPanario2013}. In particular, see Section 13.2 for enumeration and classical results, Section 14.8 for connections to $(t,m,s)$-nets, and Section 16.7 for hardware arithmetic for matrices over finite fields. In this paper, we give a new algorithm to compute determinants of all possible Hankel matrices made up from a given finite length sequence over a finite field.

We denote by $d_{i,j}$ the determinant of $\mathbf{X}_{i,j}$. By definition of
$\mathbf{X}_{i,j}$, we have for all $i$
\begin{displaymath}
d_{i,0}=1,\quad d_{i,1}=x_{i},\quad d_{i,j}=\det{\mathbf{X}_{i,j}}.
\end{displaymath}
For convenience, let $h = \big\lceil n/2\big\rceil$. We use the determinants to form a quotient-difference table \cite{Henrici1967,SloanePlouffe1995}. If $h$ is odd, the determinants
form a triangle:
\begin{displaymath}
\begin{array}{rlccccccccc}
0       &:& 1     & 1       & \ldots  & 1      & 1                   & 1                     & \ldots                 & 1         & 1    \\
1       &:& x_{0} & x_{1}   & \ldots  & \ldots & x_{h}               & \ldots                & \ldots                 & x_{n-2}   & x_{n-1}\\
2       &:&       & d_{2,1} & d_{2,2} & \ldots & \ldots              & \ldots                & d_{2,n-3}              & d_{2,n-2} &      \\
3       &:&       &         & d_{3,2} & \ldots & \ldots              & \ldots                & d_{3,n-3}              &           &      \\
4       &:&       &         &         & \ddots & \vdots              & \reflectbox{$\ddots$} &                        &           &      \\
\vdots  & &       &         &         & \ddots & \vdots              & \reflectbox{$\ddots$} &                        &           &      \\
h       &:&       &         &         &        & d_{h,h}             &                       &                        &           &
\end{array}
\end{displaymath}
If $n$ is even, then the triangle is truncated at the $h$th level where there are two elements $d_{h,h}$ and $d_{h,h+1}$. We observe that $i$ refers to columns and $j$ refers to rows.

For integers $i_0$, $i_1$, $j_0$, $j_1$ such that $i_1 > i_0$, $h \geq j_1 - j_0 > 0$, consider the set $S(i_0,i_1,j_0,j_1)=\{(i,j)\in\mathbb{N}\times\mathbb{N}\mid i_0\leq i\leq i_1,\quad j_0\leq j\leq j_1,\quad j_1-1\leq i < n-j_1+1\}$. We observe that $S$ is nonempty and may have a $k$-side polygonal shape with $3\leq k\leq 6$. Section \ref{sect:illusexa} gives two examples, one with $n=32$, $k=6$, and one with $n=81$, $k=4$. We see that the tip of the triangular table from the example with $n=32$ has length two, and therefore it yields to an hexagonal case. For a detailed explanation, see Section \ref{sect:illusexa}. If $S$ falls entirely inside the table, then $k=4$ necessarily, that is, we have a square of zeros. If $S$ overlaps with the edges of the triangular table, then $k$ may be different than $4$. We use $\partial S$ to denote the boundary of $S$. We prove in this work that zeros in a difference table are \emph{always} distributed or grouped according to $S$.

Our goal is to design a dynamic programming algorithm to fill the table that requires the least number of determinant evaluations. More precisely, if we know the first $j-1$ rows of the table, then we want to compute determinants for the $j$th row by using the least possible number of rows above the $j$th. In Section \ref{sect:results}, we establish relations among determinants $d_{i,j}$'s no matter how ${\bf x}$ is generated. Our results amplify any linear patterns that could be used to generate the coordinates of ${\bf x}$. We show that any run of zeros in the table automatically implies a run of zeros exactly below the former so that we obtain a square of zeros. Moreover, we prove identities, that we call \emph{cross shape identities}, relating determinants $d_{i,j}$'s located on a cross as explained later; those identities are based on Sylvester's identities, generalized by Bareiss \cite{Bareiss1968}, as well as Dogson's identity \cite{ABELES2014130, HornJohnson_MatAnalysis}.

It would be possible to avoid the evaluations of determinants of matrices by generalizing determinantal identities given in Conjecture \ref{conj:deteq} from Section \ref{sect:results}. More precisely, in a true random sequence of length $n$, the expected length of the maximum run of zeros is $O(\log_{2}{n})$. Therefore using the recursive nature of determinants, and especially determinants of Hankel matrices, we conjecture that the evaluations of determinants of matrices larger than about $O(\log_{2}{n})$ are not required to complete the table above which would lead to a linear time algorithm to locate the linear subsequence. Our algorithm may also be used as a statistical test to determine linearity in a pseudo-random sequence.

In Section $3$, we apply results from Section \ref{sect:results} to the case of a sequence ${\bf x}=(x_{i})_{i=1}^{n}$ that contains a linearly shifted and fed back subsequence.
\begin{definition}\label{defn:linsubseq}
Using ${\bf x}=(x_i)_{i=0}^{n-1}$ as above, let $c=(c_0,\ldots, c_{d-1})\in \mathbb{F}_{q}^{d}$ with $c_{d-1}=1$ and $d<n-1$. The sequence ${\bf x}$ contains a \emph{linear subsequence} if there are integers $s$ and $t$ with $d\leq s \leq t < n$ such that for all $s\leq \ell \leq t$ we have
\begin{displaymath}
\sum_{i=0}^{d-1}{c_{i}x_{\ell-d+i}}=0.
\end{displaymath}
\end{definition}
Indeed one of our motivations is to identify the indices $s$ and $t$ as well as to find the generating vector $c$. This relates our work to the Berlekamp-Massey algorithm. As shown later our method does not assume any upper bound on the length of $c$ or equivalently on the degree of the generating polynomial in the framework of Berlekamp-Massey.

Given a prime power $q$, $d>0$, and a sequence ${\bf x}=(x_i)_{i=0}^{\infty}$ with $x_i\in\mathbb{F}_{q}$, the Berlekamp-Massey algorithm is an iterative algorithm that finds the shortest linear feedback shift register (LFSR) that generates $\bf{ x}$. A register of size $d$ over $\mathbb{F}_{q}$ is an element of $\mathbb{F}_{q}^{d}$. More precisely, an LFSR consists of an initial register $(x_0,x_1,\ldots,x_{d-1})\in\mathbb{F}_{q}^{d}$, a non zero vector $c=(c_0,\ldots,c_{d-1})\in\mathbb{F}_{q}^{d}$ such that for $i\geq0$
\begin{equation}
(x_i,x_{i+1},\ldots, x_{i+d-2}, x_{i+d-1})\longrightarrow (x_{i+1},x_{i+2},\ldots,x_{i+d-1},\sum_{j=0}^{d-1}{c_{j}x_{i+j}}).\label{mbeqn1}
\end{equation}
The arrow in Equation (\ref{mbeqn1}) expresses the transition. The state of a system at a point in time is the content of the register. In Equation (\ref{mbeqn1}), the system transits from the state $(x_i,x_{i+1},\ldots, x_{i+d-2}, x_{i+d-1})$ to $(x_{i+1},x_{i+2},\ldots,x_{i+d-1},x_{i+d})$ where $x_{i+d}$ is a given as linear combination of the previous $x_i$'s. The content of the register at time $i+d-1$ is being fed back into the right end of it through the linear combination $\sum_{j=0}^{d-1}{c_{j}x_{i+j}}$. At time $i+d$, the register is updated to $(x_{i+1},x_{i+2},\ldots,x_{i+d-1},x_{i+d})$ where $x_{i+d}=c_{0}x_{i}+c_{1}x_{i+1}+\cdots + c_{d-1}x_{i+d-1}$.

For more information on the Berlekamp-Massey algorithm, see \cite{blahut_2003} where  interesting connections between this algorithm and the extended Euclidean algorithm are given. LaMacchia and Odlyzko \cite{lamacchiaodlyzko1990} also review how Berlekamp-Massey algorithm is used in the Wiedemann algorithm to find linear recurrences over finite fields and also show interesting connections to determinants of Hankel matrices. For more information on LFSR sequences, see \cite{Golomb_1981,golomb_gong_2005}.

We conclude the section giving the structure of the paper. Section \ref{sect:results} gives several theoretical relations among Hankel determinants that are crucial in this paper. Those relations are used in Section \ref{sect:algorithm} where we provide our algorithm to compute all determinants from Hankel matrices. Illustrative examples over $\mathbb{F}_2$ are given in Section \ref{sect:illusexa}. Due to the lack of space, experimental runs of our algorithm against a standard method of computation are given in Appendix \ref{sect:empres}. We compare our results empirically with the trivial algorithm which consists of computing determinants for each possible Hankel matrices made up from a given finite length sequence. Our new accelerated approach on a single processor is faster than the trivial algorithm on 160 processors for input sequences of length 16384 for instance as shown in Section \ref{sect:concl}.

\section{Relations among Hankel determinants}\label{sect:results}

In this section, we derive useful results to allow the computations of $d_{i,j}$ without actually computing explicitly or directly determinants of size $j$ and instead using determinants $d_{i,j'}$ with $j<j'$. Then in Section \ref{sect:algorithm}, we fill the triangular table using a dynamic programming approach. Before that, let us recall one of the results from \cite{Bareiss1968} applied to Hankel matrices and adapted to our notation. If $i$, $j$ are such that $i_0<i<i_1$, $j_0\leq j\leq j_0+(i_1-i_0-1)$, and with the convention that $d_{i,0}=1$, $d_{i,1}=x_i$, then
\begin{align}
d_{i,j}d_{i,j_{0}-1}^{j-j_{0}}&=\det\left(
\begin{array}{lcl}
d_{i,j_{0}} & \ldots & d_{i+j-j_0,j_{0}}\\
\vdots &\ddots & \vdots\\
d_{i-(j-j_0),j_{0}} & \ldots & d_{i,j_{0}}
\end{array}\right).\label{jthstepid}
\end{align}
Equation (\ref{jthstepid}) is called a $j$th-step integer preserving identity in \cite{Bareiss1968}. We call an identity like in Equation (\ref{jthstepid}) a cross shape identity because $d_{i,j}$, $d_{i,j_0}$ and $d_{i,j_{0}-1}$ are located on the vertical part of a cross, and the other non-diagonal elements of the matrix are located on the horizontal part of the aforementioned cross. A visual representation of Equation (\ref{jthstepid}) is as follow:
\[\arraycolsep=2.0pt
\begin{array}{llccccccccccc}
0      &:& 1   & 1      & \ldots & 1                                    & \ldots & 1             & \ldots                & 1           &
\ldots &   1      &      1  \\
1      &:& x_0 & x_1    & \ldots & x_{i-j+j_0}                          & \ldots & x_{i}         & \ldots                & x_{i+j-j_0} & \ldots & x_{n-2} & x_{n-1}\\
\vdots & &     & \ddots &        & \vdots                               &        & \vdots        & &                     \vdots                & &\reflectbox{$\ddots$}   &\\
\cline{8-8}
j_{0}-1&:&     &        &\ddots  &                                      &        & \multicolumn{1}{|c|}{d_{i,j_0-1}}     &                       & &\reflectbox{$\ddots$}&&\\
\cline{6-7}\cline{9-10}
j_0    &:&     &        &        & \multicolumn{1}{|c}{d_{i-j+j_0,j_0}} & \ldots & d_{i,j_0}     & \ldots                & \multicolumn{1}{c|}{d_{i+j-j_0,j_0}}                                                                                                   &             &        \\
\cline{6-7}\cline{9-10}
\vdots & &     &        &        &                                      &        &\multicolumn{1}{|c|}{\vdots}&                &             &
&        \\
\vdots & &     &        &        &                                      &        &\multicolumn{1}{|c|}{\vdots}&                &             &
&        \\
j      &:&     &        &        &                                      &        &\multicolumn{1}{|c|}{d_{i,j}}&                &&\\
\cline{8-8}
\vdots & &     &        &        &                                      &        &\vdots                       &                &             &
&        \\
\end{array}
\]
\normalsize
We come back to Equation (\ref{jthstepid}) at the end of this section with a brief explanation of its proof. Equation (\ref{jthstepid}) remains valid even if $d_{i,j_{0}-1}=0$ as pointed in \cite{Bareiss1968}.

\begin{theorem}\label{thm:carreauxdezeros}
Let $i_0< i < i_1$, and $j_0$ be such that $d_{i,j_{0}}\neq 0$, $d_{i,j_{0}+1}= 0$, $d_{i_0,j_0+1}\neq 0$, $d_{i_1,j_0+1}\neq 0$. Then with $j_1=j_0+(i_1-i_0)$ and $S(i_0,i_1,j_0,j_1)$ non-empty, we have
\begin{align*}
d_{i,j}&=0\quad\text{for all $(i,j)\in S(i_0,i_1,j_0,j_1)$},\\
d_{i,j}&\neq 0 \quad \text{for all $(i,j)\in \partial S(i_0,i_1,j_0,j_1)$.}
\end{align*}
\end{theorem}
\begin{proof}
Without loss of generality, assume that $S$ falls entirely inside the table with left and right boundaries at $(i_0,j_0)$ and $(i_1,j_0)$, respectively, and with upper and lower boundaries at $(i_0,j_0)$ and $(i_0,j_1)$, respectively. To fall entirely inside the table, one must have $2(i_{0}+1)-i_{1}\geq 0$ so that the Hankel matrix $\mathbf{X}_{i_{0},i_{1}-i_{0}-1}$ is properly defined; the number of consecutive zeros on level $j_0$ that occur between $i_0$ and $i_1$ is $i_1-i_0-1$.

Fix $i$ such that $i_0<i<i_1$ and let $j_0\leq w\leq j_0+(i_1-i_0-1)$. Then using Equation (\ref{jthstepid}) with $w=j_0+1$ as the basis for induction, we obtain that $d_{i,j_0+1}=0$, that is, we obtain the second row of zeros below the first one. For the inductive step, assume that $d_{i,w'}=0$ for $j_{0}\leq w'<w$, and rewrite Equation (\ref{jthstepid}) as
\begin{displaymath}
d_{i,w}=d_{i,j_{0}-1}^{j_{0}-w}\det\left(
\begin{array}{lcl}
d_{i,j_{0}} & \ldots & d_{i+w-j_0,j_{0}}\\
\vdots & \ddots & \vdots\\
d_{i+j_0-w,j_{0}} & \ldots & d_{i,j_{0}}
\end{array}\right).
\end{displaymath}
Therefore at least one row of the previous matrix is made only of zeros which implies the desired result.
\end{proof}

We remark that Theorem \ref{thm:carreauxdezeros} does not depend on the input sequence, and it is solely a property of determinants for Hankel matrices. If for instance the input sequence is chosen entirely at random with independent identically unbiased distributed Bernoulli random variables, then the biggest squares have average side length $O(\log_{2}{n})$ which is the expected length of the longest run of zeros in a random sequence of Bernoulli random variables with length $n$.

Given a square matrix $\mathbf{X}$ of size $\ell\times \ell$, we consider its sub-matrix $\mathbf{C}$ of size $(\ell-2)\times (\ell-2)$ located in the center $\mathbf{X}$, and its $4$ sub-matrices $\mathbf{N}$, $\mathbf{S}$, $\mathbf{E}$ and $\mathbf{W}$ of size $(\ell-1)\times(\ell-1)$ located in the top left, bottom right, top right and bottom left of $\mathbf{X}$, respectively. In other words let
\begin{align*}
\mathbf{X}&=\left(\begin{array}{c|c|c}
x_{1,1} & \ldots & x_{1,\ell}\\
\hline
\vdots & \mathbf{C} & \vdots\\
\hline
x_{\ell,1} & \ldots & x_{\ell,\ell}
\end{array}\right)\\
&=\left(
\begin{array}{c|r}
    \mathbf{\text{\Large{N}}} & \begin{matrix} x_{1,\ell} \\ \vdots\end{matrix} \\ \hline
    \begin{matrix} x_{\ell,1} & \ldots \end{matrix} & x_{\ell,\ell}
\end{array}
\right)=\left(
\begin{array}{l|c}
    x_{1,1} & \begin{matrix} \ldots & x_{1,\ell}\end{matrix} \\ \hline
    \begin{matrix} \vdots \\ x_{\ell,1} \end{matrix} & \mathbf{\text{\Large{S}}}
    \end{array}
\right)\\
&=\left(\begin{array}{c|r}
    \begin{matrix} x_{1,1} &\ldots \end{matrix} & x_{1,\ell} \\ \hline
    \mathbf{\text{\Large{W}}} & \begin{matrix} \vdots \\ x_{\ell,\ell} \end{matrix}
    \end{array}
\right)=\left(
\begin{array}{l|c}
    \begin{matrix} x_{1,1} \\ \vdots \end{matrix} & \mathbf{\text{\Large{E}}} \\ \hline
    x_{\ell,1} & \begin{matrix} \ldots & x_{\ell,\ell}\end{matrix}
\end{array}
\right).
\end{align*}
Then we have Dodgson's identity (see \cite{ABELES2014130}, or page $29$ of \cite{HornJohnson_MatAnalysis}):
\begin{align}
\det(\mathbf{X})\det(\mathbf{C})&=\det(\mathbf{N})\det(\mathbf{S})-\det(\mathbf{E})\det(\mathbf{W})\label{machinNSEW}.
\end{align}
If the entry $x_{\ell,\ell}$ is an unknown and all other elements of $\mathbf{X}$ are known, then, for some $\alpha,\beta\in\mathbb{F}_{q}$, we have that
\begin{align}
\big(x_{\ell,\ell}\det(\mathbf{N})+\alpha\big)\det(\mathbf{C})=\det(\mathbf{N})\big(x_{\ell,\ell}\det(\mathbf{C})+\beta\big)-\det(\mathbf{E})\det(\mathbf{W}).\label{facto1}
\end{align}
Equation (\ref{facto1}) implies that $x_{\ell,\ell}$ cannot be determined if $\det(\mathbf{N})=0$ or $\det(\mathbf{C})=0$. This simply implies that $x_{\ell,\ell}$ cannot be determined from a determinantal equation of the type obtained by Dodgson's identity.

We now derive a useful identity using Equation (\ref{machinNSEW}) which can also be proved using results from \cite{Bareiss1968}.
\begin{proposition}[North-South-East-West]\label{prop:1step}
For all $(i,j)$ such that $i-j+1\geq 0$, and $2\leq j\leq\big\lceil n/2 \big\rceil$ the following identity is true:
\begin{displaymath}
d_{i,j}d_{i,j-2}=d_{i,j-1}^2-d_{i+1,j-1}d_{i-1,j-1}.
\end{displaymath}
\end{proposition}
\begin{proof}
Apply Equation (\ref{machinNSEW}) on the matrix $\mathbf{X}_{i,j}$ given from (\ref{matrepxij}) where $\det(\mathbf{W})=d_{i-1,j-1}$, $\det(\mathbf{E})=d_{i+1,j-1}$, $\det(\mathbf{N})=d_{i,j-1}$, $\det(\mathbf{S})=d_{i,j-1}$, and $\det(\mathbf{C})=d_{i,j-2}$.
\end{proof}

We observe that Proposition \ref{prop:1step} is reminiscent to the North-South-East-West identity \cite{WolframQuoDif} for quotient-difference table. Proposition \ref{prop:1step} is similar to the $1$st-order step integer preserving relation from \cite{Bareiss1968} with a much easier proof. The condition $d_{i,j-2}\neq 0$ is not required as explained in \cite{Bareiss1968}, or as it follows directly from Equation (\ref{machinNSEW}), but it matters for our dynamic programming method since we cannot determine $d_{i,j}$ if $d_{i,j-2}=0$ using the table information from the $(j-1)$th and $(j-2)$th rows.

In order to accelerate the computation of determinants within a dynamical programming approach, we must ensure that $d_{i,j_{0}-1}\neq 0$ from Equation (\ref{jthstepid}). For that we have the next theorem.
\begin{theorem}\label{thm:noese}
For all $(i,j)$ such that $i-j+1\geq 0$, and $j_0\leq j\leq\big\lceil n/2 \big\rceil$, if
\begin{displaymath}
d_{i,j_{0}-1}\neq 0,\quad d_{i,k}=0\text{ for $j_{0}\leq k\leq j-1$},
\end{displaymath}
then
\begin{align*}
d_{i,j}&=d_{i,j_{0}-1}^{j_{0}-j}\det\left(
\begin{array}{lcl}
d_{i,j_{0}} & \ldots & d_{i+j-j_0,j_{0}}\\
\vdots &\ddots & \vdots\\
d_{i-(j-j_0),j_{0}} & \ldots & d_{i,j_{0}}
\end{array}\right).
\end{align*}
\end{theorem}
\begin{proof}
Before starting, for a fixed position $i$ and for any size $j_0$, we observe that the value $j-j_{0}$ expresses the depth of singularity, that is, the number of zeros below the non-zero cell indexed by $(i,j_0-1)$. The depth of singularity also relates to the concentration of zeros aligned horizontally around the cell $(i,j_0-1)$. By concentration of zeros, we mean the length of a run of consecutive zeros.

\textbf{$(j=j_{0}+1)$th step:} Suppose that $d_{i,j_{0}-1}\neq 0$ and $d_{i,j_{0}}=0$. If $d_{i,j_{0}+1}=0$, then there is at least one zero to the left or to the right of $(i,j_{0})$ or both. Indeed Proposition \ref{prop:1step} entails that $d_{i,j_{0}+1}d_{i,j_{0}-1}=d_{i,j_{0}}^{2}-d_{i+1,j_{0}}d_{i-1,j_{0}}$ which, in this case, is equivalent to $d_{i,j_{0}+1}d_{i,j_{0}-1}=-d_{i+1,j_{0}}d_{i-1,j_{0}}$ from which we infer that either $d_{i-1,j_{0}}=0$ or $d_{i+1,j_{0}}=0$ whenever $d_{i,j_{0}+1}=0$. So there is qualitatively speaking a small concentration of zeros aligned horizontally around the cell $(i,j_0)$.

\textbf{$(j=j_{0}+2)$th step:} Now suppose that $d_{i,j_{0}-1}\neq 0$, $d_{i,j_{0}}=d_{i,j_{0}+1}=0$ and write
\begin{displaymath}
0=d_{i,j_{0}+1}d_{i,j_{0}-1}^{2}=\det\left(\begin{array}{ccc}
0 & d_{i+1,j_{0}} & d_{i+2,j_{0}} \\
d_{i-1,j_{0}} & 0 & d_{i+1,j_{0}}\\
d_{i-2,j_{0}} & d_{i-1,j_{0}} & 0
\end{array}\right).
\end{displaymath}
By the $(j=j_{0}+1)$th-step, if $d_{i+1,j_{0}}=0$, then $0=d_{i,j_{0}+1}d_{i,j_{0}-1}^{2}=d_{i-1,j_{0}}^{2}d_{i+2,j_{0}}$ from which either $d_{i-1,j_{0}}=0$ or $d_{i+2,j_{0}}=0$; if $d_{i-1,j_{0}}=0$, then $0=d_{i,j_{0}+1}d_{i,j_{0}-1}^{2}$ $=d_{i+1,j_{0}}^{2}d_{i-2,j_{0}}$ from which either $d_{i+1,j_{0}}=0$ or $d_{i-2,j_{0}}=0$. Therefore we conclude that $d_{i-1,j_{0}}=d_{i+1,j_{0}}=0$ as well. The horizontal part of the cross contains therefore a higher concentration of zeros around $d_{i,j_{0}}$ with respect to the previous step. We cannot conclude at this moment that $d_{i+1,j_{0}}=0=d_{i+2,j_{0}}$ or $d_{i-1,j_{0}}=0=d_{i-2,j_{0}}$ without further adding deeper singularities.

\textbf{$(j=j_{0}+3)$th step:} Now suppose that $d_{i,j_{0}-1}\neq 0$ and $d_{i,j_{0}}=d_{i,j_{0}+1}=d_{i,j_{0}+2}=0$ and write
\begin{align}
0&=d_{i,j_{0}+2}d_{i,j_{0}-1}^{3}=\det\left(\begin{array}{cccc}
 0 & d_{i+1,j_{0}} & d_{i+2,j_{0}} & d_{i+3,j_{0}} \\
d_{i-1,j_{0}} & 0 & d_{i+1,j_{0}} & d_{i+2,j_{0}}\\
d_{i-2,j_{0}} & d_{i-1,j_{0}} & 0 & d_{i+1,j_{0}}\\
d_{i-3,j_{0}} & d_{i-2,j_{0}} & d_{i-1,j_{0}} & 0
\end{array}\right).\label{eq:flatus}
\end{align}
From the $(j=j_{0}+2)$th step, $d_{i-1,j_{0}}=0=d_{i+1,j_{0}}$, and Equation (\ref{eq:flatus}) is equivalent to
\begin{align*}
0&=d_{i,j_{0}+2}d_{i,j_{0}-1}^{3} 
=\det\left(\begin{array}{cccc}
 0 & 0 & d_{i+2,j_{0}} & d_{i+3,j_{0}} \\
 0 & 0 & 0 & d_{i+2,j_{0}}\\
d_{i-2,j_{0}} & 0 & 0 & 0\\
d_{i-3,j_{0}} & d_{i-2,j_{0}} & 0 & 0
\end{array}\right) \\
&=d_{i+2,j_{0}}^{2}d_{i-2,j_{0}}^{2},
\end{align*}
so that either $d_{i-2,j_{0}}=0$ or $d_{i+2,j_{0}}=0$. With the knowledge of the $(j_{0}+2)$th-step, we can safely conclude that either $d_{i+1,j_{0}}=0=d_{i+2,j_{0}}$ or $d_{i-2,j_{0}}=0=d_{i-1,j_{0}}$. Hence the concentration of zeros increases on the $j_{0}$th row with respect to the previous steps. We observe that the previous determinant has at least one row with $3$ consecutive zeros. The position of a run of zeros from one row to the following is shifted cyclically by one position.


The process stops when we can no longer add deeper singularity,
that is, we stop for the smallest index $j>j_{0}$ such that
$d_{i,j}\neq 0$. When such index $j$ is found, then we can no
longer deduce zero determinants on the horizontal part of the
cross.

\textbf{$j$th step:} Assume that $d_{i,j_{0}-1}\neq 0$ and $d_{i,k}=0$ for $j_{0}\leq k\leq j-1$, and now assume $d_{i,j}\neq 0$. The matrix to consider at this step has size $(j+1)\times(j+1)$. Thus at this current $j$th-step, we can find $d_{i,j}d_{i,j_{0}-1}^{j-j_{0}} \neq 0$. We observe that we hit the boundaries of a square of zeros. At the following $(j+1)$th step, all rows would contain at least $2$ non-zero elements or equivalently there would not be a row with at least $j+1$ consecutive zeros. Thus $d_{i,j+1}d_{i,j_{0}-1}^{j+1-j_{0}}$ would be the sum of at least two products and it would become impossible to correctly deduce the values of the determinants.
\end{proof}

Now we work for our next result, a partially proved conjecture. Given valid indices $i,j$ for the column and the row of the triangular table of determinants, and $2k\leq j$ let $\mathbf{G}_{i,j,k}$ be a matrix of size $(k+1)\times (k+1)$ defined as
\begin{align}
\mathbf{G}_{i,j,k}&=\left(\begin{array}{cccc}
d_{i,j-2k} & d_{i+1,j-2k+1} & \ldots & d_{i+k,j-k}\\
d_{i-1,j-2k+1} & d_{i,j-2k+2} & \ldots & d_{i+k-1,j-k+1}\\
\vdots &\vdots &\ddots & \vdots\\
d_{i-k,j-k} & d_{i-k+1,j-k+1} & \ldots & d_{i,j}
\end{array}\right).
\end{align}
In other words, for $0\leq r,c \leq k$, the entry of $\mathbf{G}_{i,j,k}$ located on the $r$th row and $c$th column is given by $d_{i-r+c,j-2k+r+c}$. The pair $(i-r+c,j-2k+r+c)$ indexing an element of $\mathbf{G}_{i,j,k}$ is the intersection of two perpendicular lines. The intersection of a group of $k$ parallel lines intersecting perpendicularly another group of $k$ parallel lines as it might be easier to see with the following representation by drawing $k$ lines with slope $\pi/4$ and separated at distance $\sqrt{2}$ intersecting $k$ other lines with slope $3\pi/4$ also at distance $\sqrt{2}$ of each other:

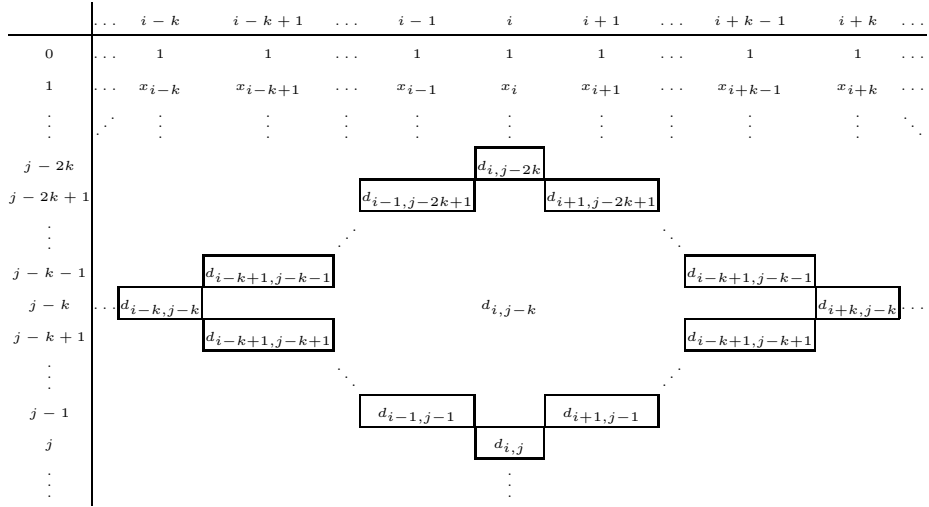
\tiny
\[\arraycolsep=0.2pt\def\arraystretch{2.0}
\begin{array}{c|ccccccccccccc}
         & \ldots & i-k     & i-k+1     & \ldots & i-1     & i   & i+1    & \ldots  & i+k-1     & i+k     & \ldots\\ \hline
0      & \ldots & 1       & 1         & \ldots & 1       & 1   & 1      & \ldots  & 1         & 1       & \ldots\\
1      & \ldots & x_{i-k} & x_{i-k+1} & \ldots & x_{i-1} & x_{i}&x_{i+1}& \ldots  & x_{i+k-1} & x_{i+k} & \ldots\\
\vdots &\reflectbox{$\ddots$}& \vdots     & \vdots &  \vdots     & \vdots & \vdots    & \vdots& \vdots   & \vdots  &\vdots&\ddots\\
\cline{7-7} j-2k   &       &            &             &        &           &\multicolumn{1}{|c|}{d_{i,j-2k}}&     &         &             &           &       \\
\cline{6-8}j-2k+1 &       &            &             &        &\multicolumn{1}{|c|}{d_{i-1,j-2k+1}}&  &\multicolumn{1}{|c|}{d_{i+1,j-2k+1}}&   &             &           & \\
\cline{6-6}\cline{8-8}
\vdots &       &            &             &\reflectbox{$\ddots$}&    &         &   &\ddots         &             &           & \\[-0.5mm]
\cline{4-4}\cline{10-10}j-k-1  &       &            &\multicolumn{1}{|c|}{d_{i-k+1,j-k-1}}&    &           &      &          &         &\multicolumn{1}{|c|}{d_{i-k+1,j-k-1}}&       & \\
\cline{4-4}\cline{10-10}
\cline{3-3}\cline{11-11}j-k    & \ldots&\multicolumn{1}{|c|}{d_{i-k,j-k}} &           &        &           & d_{i,j-k}      &          &         &               &\multicolumn{1}{|c|}{d_{i+k,j-k}}    &\ldots\\
\cline{3-3}\cline{11-11}\cline{4-4}\cline{10-10} j-k+1  &       &            &\multicolumn{1}{|c|}{d_{i-k+1,j-k+1}}&    &           &       &          &         &\multicolumn{1}{|c|}{d_{i-k+1,j-k+1}}&       & \\
\cline{4-4}\cline{10-10}\vdots &       &            &             &\ddots&           &       &          &\reflectbox{$\ddots$}         &             &           &\\[-0.5mm]
\cline{6-6}\cline{8-8}j-1    &       &            &             &        &\multicolumn{1}{|c|}{d_{i-1,j-1}}&  &\multicolumn{1}{|c|}{d_{i+1,j-1}}&   &             &           & \\
\cline{6-8}j      &       &            &             &        &           &\multicolumn{1}{|c|}{d_{i,j}}&     &         &             &           &       \\
\cline{7-7}
\vdots       &       &            &             &        &   & \vdots    &  &&&&\\
\end{array}
\]
\normalsize

If the information about the determinants $d_{i,j'}$, $0\leq j' \leq j-1$, is known, then we might hope to solve a determinantal equation like $\det(\mathbf{G}_{i,j,k})=g$ for some $g\in\mathbb{F}_{q}$ for the unknown $d_{i,j}$ located in the bottom right corner of $\mathbf{G}_{i,j,k}$.

We may sometimes abuse the language to denote the index $(i,j)$ or the value indexed by $(i,j)$ which is $d_{i,j}$. It is very convenient to refer to $k$ as a radius of an $\ell_{1}$-ball centered around $d_{i,j-k}$ or more precisely around the index $(i,j-k)$. An $\ell_{1}$-ball is a square grid. The grid can be seen as the intersection of the two families of parallel lines and each family perpendicular to each other as mentioned previously. The indices obtained by the intersection of the two families are used to define $\mathbf{G}_{i,j,k}$. We refer the neighbourhood around $(i,j-k)$, which is the center of the grid, as the $\ell_{1}$-ball of radius $k$. If $k$ is even, then the center $(i,j-k)$ is deleted. If $k$ is odd, then the center is part of the ball.
If $\det(\mathbf{G}_{i,j,k})=0$, then there is a local linear dependency around $(i,j-k)$. From Dodgson's identity, $\mathbf{G}_{i,j,k-2}$ plays the role of the center which can be seen as the interior of the neighbourhood of $(i,j-k)$.

We postulate the following conjecture about the local linear dependency or more precisely about the minimum amount of information required to determine $d_{i,j}$ assuming the table is known up to the $(j-1)$th level, inclusively.
\begin{conjecture}\label{conj:deteq}
For $2\leq k\leq 6$, $j\geq 2k$, and $j-1\leq i\leq n-j$, where $n$ is the length of the sequence, the smallest radius $k$ for which $\det{\mathbf{G}}_{i,j,k}=0$ is the smallest value $k$ for which $\det{\mathbf{G}}_{i,j,k-1}\neq 0$. Assuming $\det{\mathbf{G}}_{i,j,k'}=0$ for $2\leq k'\leq 6$, then 
we have $\det{\mathbf{G}}_{i,j,7}=0$ if and only $\det{\mathbf{G}}_{i,j,6}\neq 0$ and $d_{i,j-7}=0$.
\end{conjecture}

We recall from linear algebra that the condition $\det(\mathbf{G}_{i,j,k-1})\neq 0$ is necessary and sufficient for the uniqueness to the solution of the linear equation $\det(\mathbf{G}_{i,j,k}) = 0$ with $d_{i,j}$ as unknown for all value of $k$ and $j\geq 2k$.

We verified the previous conjecture by comparing $d_{i,j}$ obtained from solving the corresponding determinantal equation with the value obtained from the trivial algorithm.
Given that we never found any counter-example to Conjecture \ref{conj:deteq}, we decided to include its algorithmic flavour, that is Algorithm \ref{alg:conjecture}, in our dynamic method given in Algorithm \ref{alg:dynamicprog}.

\begin{breakablealgorithm}\label{alg:conjecture}
\caption{Growing an $\ell_{1}$-metric ball and solving for the unknown}
\begin{algorithmic}[1]
\raggedright
\Require Integer $n>0$, $2\leq k\leq 7$, $j\geq 2k$, and $j-1\leq i\leq n-j$
\If{$k\leq 6$}
\For{$k'=2$ to $k$}
\If{$\det(\mathbf{G}_{i,j,k'-1})\neq 0$}
\State{Solve $\det(\mathbf{G}_{i,j,k'})=0$ with $d_{i,j}$ as unknown}
\State{\textbf{Return} $d_{i,j}$}
\EndIf
\EndFor
\ElsIf{$k=7$, $\det(\mathbf{G}_{i,j,6})\neq 0$, $d_{i,j-7}=0$}
\State{Solve $\det(\mathbf{G}_{i,j,7})=0$ with $d_{i,j}$ as unknown}
\State{\textbf{Return} $d_{i,j}$}
\Else
\State{$k$ out of range}\Comment{We need more research for larger radius $k$.}
\EndIf
\end{algorithmic}
\end{breakablealgorithm}



We finish by briefly explaining Equation (\ref{jthstepid}) as we promised at the beginning of this section. Given a square matrix $A$ of size $h\times h$, with $h=\lceil n/2\rceil$, we introduce the notation from \cite{Bareiss1968}
\begin{displaymath}
a_{r,c}^{(k)}=\det\left(
\begin{array}{ccccc}
a_{0,0} & a_{0,1} & \ldots & a_{0,k-1} & a_{0,c}\\
a_{1,0} & a_{1,1} & \ldots & a_{1,k-1} & a_{1,c}\\
\vdots  & \vdots  & \ddots & \vdots    & \vdots\\
a_{k-1,0} & a_{k-1,1} & \ldots & a_{k-1,k-1} & a_{k-1,c}\\
a_{r,0} & a_{r,1} & \ldots & a_{r,k-1} & a_{r,c}
\end{array}
\right)\quad\text{for $k\leq r,c\leq h$.}
\end{displaymath}
Clearly $a_{r,c}^{(k)}$ is the determinant of a $(k+1)\times(k+1)$ matrix. We observe that the principal minors of $A$ are $a_{k,k}^{(k)}$. In \cite{Bareiss1968}, it is shown that
\begin{align}
a_{r,c}^{(k)}&=\frac{1}{\big(a_{\ell,\ell}^{(\ell-1)}\big)^{k-\ell}}
\det\left(\begin{array}{llll}
a_{\ell,\ell}^{(\ell)} & \ldots & a_{\ell,k-1}^{(\ell)} & a_{\ell,c}^{(\ell)} \\
\vdots & \ddots & \vdots & \vdots\\
a_{k-1,\ell}^{(\ell)} & \ldots & a_{k-1,k-1}^{(\ell)} & a_{k-1,c}^{(\ell)} \\
a_{r,\ell}^{(\ell)} & \ldots & a_{r,k-1}^{(\ell)} & a_{r,c}^{(\ell)}
\end{array}
\right)\text{for $0<\ell< k$.}\label{gugussemachin}
\end{align}
So $a_{r,c}^{(k)}$ is also the determinant of a $(k-\ell+1)\times (k-\ell+1)$ matrix of determinants. The left side of Equation (\ref{gugussemachin}) does not depend on $\ell$. Let us concentrate on the principal minors when $r=c=k$, and substitute $\ell=j_0$ and $k=j$ in Equation (\ref{gugussemachin}) to get that
\begin{align*}
a_{j,j}^{(j)}&=\frac{1}{\big(a_{j_0,j_0}^{(j_0-1)}\big)^{j-j_0}}
\det\left(\begin{array}{llll}
a_{j_0,j_0}^{(j_0)} & \ldots & a_{j_0,j-1}^{(j_0)} & a_{j_0,j}^{(j_0)} \\
\vdots & \ddots & \vdots & \vdots\\
a_{j-1,j_0}^{(j_0)} & \ldots & a_{j-1,j-1}^{(j_0)} & a_{j-1,j}^{(j_0)} \\
a_{j,j_0}^{(j_0)} & \ldots & a_{j,j-1}^{(j_0)} & a_{j,j}^{(j_0)}
\end{array}
\right)\quad\text{for $0<j_0< j$.}
\end{align*}
Entries (determinants) inside of the previous matrix are those $d_{j-j_0,j_0}$'s introduced at the beginning of this section. We observe that we can shift $j-j_0$ by any quantity modulo $n$ and therefore get Equation (\ref{jthstepid}).

\section{Algorithm to compute determinants of Hankel matrices over finite fields}\label{sect:algorithm}

In the following algorithm, the symbol $j$ indexes the rows of the table and $j$ is the size of the Hankel matrix $\mathbf{X}_{i,j}$ as in the introduction. The symbol $i$ indexes the columns and is related to the position in the input vector ${\bf x}$ from where we build the Hankel matrix $\mathbf{X}_{i,j}$. We use the symbol $M$ to denote the dynamic table under consideration, and $M[j][i]$ stands for $d_{i,j}=\det{\mathbf{X}_{i,j}}$. The table grows from top to bottom by considering the smallest possible Hankel matrices to the largest one if $n$ is odd or the largest two if $n$ is even. For a given matrix of size $j$ or equivalently a given row $j$ of $M$, the algorithm sweeps from the left to the right using the input vector ${\bf x}$ in order to consider all possible Hankel matrices of size $j$. We use also $M[j][\cdot]$ to refer to the $j$th row of the table $M$.

\begin{breakablealgorithm}\label{alg:dynamicprog}
\caption{Computing determinants for all possible Hankel matrices made up from a sequence ${\bf x}\in\mathbb{F}_{q}^{n}$}
\begin{algorithmic}[1]
\raggedright
\Require Integer $n>0$ and vector ${\bf x}\in\mathbb{F}_{q}^{n}$.
\Ensure Triangular table $M$.
\State{$h\leftarrow\lceil{n}/{2}\rceil$}
\State{$M\leftarrow \emptyset$\Comment{Allocate space for $M$ with base $h$ and width $n$.}\label{algL1}}
\For{$i=0$ to $n-1$}\Comment{Initialize first two rows $M$.}
\State{$M[0][i]\leftarrow 1$}
\State{$M[1][i]\leftarrow x_{i}$}
\EndFor
\For{$j=2$ to $h$}\label{algloopj}
\State{Find new squares of zeros$\quad$ \Comment{Use Theorem \ref{thm:carreauxdezeros} with the knowledge\\ \hfill of rows $M[j-1][\cdot]$ and $M[j-2][\cdot]$.}}
\For{$i=j-1$ to $n-j$\label{algSweepPhaseStart}} \Comment{Loop is parallelized.}
\If{$M[j][i]$ has not been yet evaluated\label{algL2}}
\If{$M[j-2][i]\neq 0$}
\State{Compute $M[j][i]$ using Proposition \ref{prop:1step}}
\ElsIf{Conditions for Conjecture \ref{conj:deteq}}
\State{Compute $M[j][i]$ accordingly with Algorithm \ref{alg:conjecture}}
\ElsIf{Conditions for Theorem \ref{thm:noese}}
\State{Compute $M[j][i]$ accordingly}
\Else
\State{Compute $M[j][i]$ explicitly from its definition}
\EndIf
\EndIf
\EndFor
\EndFor
\end{algorithmic}
\end{breakablealgorithm}

Based on the results from Section \ref{sect:results}, the algorithm correctly terminates. We note that an auxiliary table can be maintained to flag entries of $M$ that were computed or not. Given $j$ from line (\ref{algloopj}), to find squares of zeros using $M[j-1][\cdot]$ and $M[j-2][\cdot]$, we look for consecutive non zero elements between two indices, say $i_0$ and $i_1$ (including $M[j-2][i_0]\neq 0$ and $M[j-2][i_1]\neq 0$) at level $j-2$, then check for $M[j-1][i_0]\neq 0$ followed by zero elements until $M[j-1][i_1]\neq 0$; the procedure begins with $i_0=j-1$ and if $i_1$ is found to be the right upper corner, then a square is filled, and the procedure continues from $i_1$ until reaching $n-j$. Once the squares are filled, the remaining elements on a given row must be evaluated. The goal is to use as few as possible knowledge from the previous rows by using Proposition \ref{prop:1step}, Theorem \ref{thm:noese}, and Conjecture \ref{conj:deteq}. If none of the previous results applied, then we revert to the trivial and expensive evaluation.

We end this section by explaining briefly how to find the generating vector of a linear subsequence. Once an unusual long run of zeros is found on a row of the table, we stop the computations of determinants since actually there is no need to further complete the table. Indeed, all the knowledge we need to build the adjugate, in order to invert a Hankel matrix connecting the generating vector to a part of the original sequence, is located on the previous rows that had been already computed.

\section{Illustrative visual examples}\label{sect:illusexa}

In this section, we give two examples illustrating our new results over $\mathbb{F}_{2}$. For the first example, we generate a sequence of length $32$ indexed from left to right starting with index $0$, ending with index $31$, and which is given by
\begin{center}
\begin{BVerbatim}[commandchars=\\\{\}]
\textcolor{red}{01010110}\textcolor{green}{1001110100111010}\textcolor{red}{11101110}
\end{BVerbatim}
\end{center}
Red color represents the prefix and the postfix that are generated randomly. Green color represents the middle linear substring. The big square of zeros due to the linearity of the middle string is in blue color. The sequence is used to initialize the table so it is identical to the row indexed by $1$ below. Row $0$ contains only unit elements. The generating vector is $(1, 0, 1, 1)=(c_0,c_1,c_2,c_3)$. We note that as mentioned previously, $c_3=1$ to ensure the vector is not trivial. The leftmost index of linear subsequence is $8$, that is $i_0=8$, and so the generated random prefix is $01010110=x_0x_1\cdots x_7$. The rightmost index of the linear subsequence is $24$, and the generated random postfix is $11101110=x_{24}\cdots x_{31}$.

The middle linearly substring is given by
$1001110100111010=x_{8}x_{9}\cdots x_{23}$,
and is generated linearly from the prefix string:
$c_3x_8+c_2x_7+c_1x_6+c_0x_5=0$ 
implies $x_8 = c_2x_7+c_1x_6+c_0x_5 = x_7+x_5 = 0 + 1 = 1$,
$x_9 = x_8+x_6=1+1=0,
x_{10} = x_9+x_7 = 0 + 0 =0$, and so on.

Since the generating vector has length $4$, then the row at which appears a long run of zeros is on the row indexed by $4$. The shape of $S$ is hexagonal, and the values of $j_1$ varies with those of the positional indices $i$. The value $j_0=4$.
\begin{center}
\scriptsize
\begin{BVerbatim}[commandchars=\\\{\}]
 0 :11111111111111111111111111111111
 1 :\textcolor{red}{01010110}\textcolor{green}{1001110100111010}\textcolor{red}{11101110}
 2 : 111111110010111001011111011101
 3 :  0010011111111111111101111011
 4 :   01001\textcolor{blue}{0000000000000}11101111
 5 :    1111\textcolor{blue}{0000000000000}1111100
 6 :     011\textcolor{blue}{0000000000000}100110
 7 :      11\textcolor{blue}{0000000000000}10011
 8 :       1\textcolor{blue}{0000000000000}1111
 9 :        \textcolor{blue}{0000000000000}100
10 :         \textcolor{blue}{000000000000}10
11 :          \textcolor{blue}{00000000000}1
12 :           \textcolor{blue}{0000000000}
13 :            \textcolor{blue}{00000000}
14 :             \textcolor{blue}{000000}
15 :              \textcolor{blue}{0000}
16 :               \textcolor{blue}{00}
\end{BVerbatim}
\end{center}
\normalsize

For the second example, we generate a sequence of length $81$ indexed from left to right starting with index $0$, ending with index $80$, and which is given by
\footnotesize
\begin{BVerbatim}[commandchars=\\\{\}]
\textcolor{red}{101100000010101111011010110101}\textcolor{green}{10010001111010110010}\textcolor{red}{0010101111011001100110000000100}
\end{BVerbatim}

\vspace{0.5cm}

\tiny
\begin{center}
\begin{BVerbatim}[commandchars=\\\{\}]
 0 :111111111111111111111111111111111111111111111111111111111111111111111111111111111
 1 :\textcolor{red}{101100000010101111011010110101}\textcolor{green}{10010001111010110010}\textcolor{red}{0010101111011001100110000000100}
 2 : 1110000001111100111111111111110010001001111110010001111100111100110011000000010
 3 :  01000000100110010001010010101111000100110101111000100110010011111111100000001
 4 :   100000010011111000111001111111111111111111111111110011111001000000010000000
 5 :    000000111101010001111111\textcolor{blue}{000000000000000000000}1101111010111100000001000000
 6 :     00000101111111111000001\textcolor{blue}{000000000000000000000}111111111111010000000100000
 7 :      1111111000001111000001\textcolor{blue}{000000000000000000000}10000000010111000000010000
 8 :       001011000001001000001\textcolor{blue}{000000000000000000000}1000000001111100000001111
 9 :        01111000001001000001\textcolor{blue}{000000000000000000000}100000000100010000000100
10 :         1001000001111000001\textcolor{blue}{000000000000000000000}10000000010001000000010
11 :          001000001111111111\textcolor{blue}{000000000000000000000}1000000001000111111111
12 :           11111111000000001\textcolor{blue}{000000000000000000000}100000000111111010010
13 :            0000011000000001\textcolor{blue}{000000000000000000000}10000000011110111001
14 :             000011000000001\textcolor{blue}{000000000000000000000}1000000001001111111
15 :              00011000000001\textcolor{blue}{000000000000000000000}111111111100110001
16 :               0011000000001\textcolor{blue}{000000000000000000000}11111101111111000
17 :                011000000001\textcolor{blue}{000000000000000000000}1000011100100100
18 :                 11000000001\textcolor{blue}{000000000000000000000}100001010010011
19 :                  1000000001\textcolor{blue}{000000000000000000000}10000111111111
20 :                   111111111\textcolor{blue}{000000000000000000000}1000010010011
21 :                    11110001\textcolor{blue}{000000000000000000000}111111001001
22 :                     0010001\textcolor{blue}{000000000000000000000}10101111111
23 :                      010001\textcolor{blue}{000000000000000000000}1111100000
24 :                       11111\textcolor{blue}{000000000000000000000}100110000
25 :                        0011\textcolor{blue}{000000000000000000000}10011000
26 :                         0111111111111111111111111111100
27 :                          10111100011010001101000001010
28 :                           110010001111000111100000111
29 :                            1001000101100010110000011
30 :                             11111111111111111000001
31 :                              001110000000000100000
32 :                               0101000000000011111
33 :                                11100000000001010
34 :                                 110000000000111
35 :                                  1000000000010
36 :                                   00000000001
37 :                                    000000000
38 :                                     0000000
39 :                                      00000
40 :                                       000
41 :                                        1
\end{BVerbatim}
\end{center}

\normalsize
In this case the generating vector is $(1,0,0,1,1)=(c_0,c_1,c_2,c_3,c_4)$.
The leftmost index of the linear subsequence is $30$, and the rightmost index of linear subsequence is $50$. The shape of $S$ is a square and $i_0=30$, $i_1=50$, $j_0=5$, \emph{and} $j_1=j_0+(i_1-i_0)=25$ as it can be seen as well from the visual aid. The prefix random string is given by $x_0x_1\cdots x_{29}=101100000010101111011010110101$. The postfix random string is $x_{50}x_{51}\cdots x_{80}=0010101111011001100110000000100$. The middle linear substring is given by $x_{30}\cdots x_{49}=10010001111010110010$.

\section{Conclusion and further work}\label{sect:concl}

We believe that there are still more relations to be found and to be coded in order to avoid the computation of determinants, and this is currently under study. An ultimate goal is to get rid entirely of the evaluations of large determinants by proving and generalizing Conjecture \ref{conj:deteq}.
How would the linear dependency vanish as the radius gets larger or synonymously how far does it propagate around the center?
Can we further enlarge the radius by adding new conditions for $k>7$?
%

It would be interesting to adapt our algorithm to output the generating vector and compare it to efficient implementations of the Berlekamp-Massey algorithm. We would need to stop at the level containing a long run of zeros and use the information of the row preceding this one to solve efficiently the linear system for the generating vector using adjugate matrices.

It is known that Berlekamp-Massey algorithm is virtually the same as the extended Euclidean algorithm for polynomials over finite field. Could we find a similar equivalence to our algorithm for problems involving Bezout identities that express linear dependencies among elements in fields?



Our dynamic approach can be easily adapted to multiple and combined linear feedback shift registers. Further research also includes to analyze the case of non-linear feedback shift register by linearizing the generator; more precisely, linearizing a non-linear boolean feedback function pertains to add constraints which are reflected in the determinant identities.

\bibliographystyle{plain}
\newcommand{\SortNoop}[1]{}

\appendix
\noindent
{\Large {\bf Appendix: Run times and distribution of counts}}\label{sect:empres}

\vspace{0.5cm}

In order to compare in practice the running times between the trivial method and our new method, we generate sequences of length $n$ that we linearly filled. In order to compute determinants of large Hankel matrices whenever necessary, we do not use the Levinson-Durbin algorithm \cite{BojanczykETAL1995}, \cite{Musicus1988} that can be adapted to Hankel matrices instead of Toeplitz matrices. We created an extremely fast C/C++ low-level module to compute determinants over $\mathbb{F}_{2}$ in order to do not rely on any external libraries. Our module to compute determinants over $\mathbf{F}_{2}$ is quite faster than NTL; it however only applies to binary matrices. We recall that one of our future goals is to avoid such computation of determinants of large matrices, and only use local information.
We look at typical worst-case instances when the linear subsequence is ``buried'' between two long random sequences serving as a prefix and a postfix. The prefix random string together with the generator vector are used to built the linear subsequence in the middle. The prefix random string must be at least as long as the length of the generating vector to be used as initial data.
We also consider typical easy-case instances when there is no random postfix sequence and when the length of the random prefix sequence is the same as the generating vector.

For our accelerated dynamic algorithm, we give the distribution of counts of the number of times, with respect to the number of entries in the table, that we branch to Proposition \ref{prop:1step}, Theorems \ref{thm:carreauxdezeros} or \ref{thm:noese}, Conjecture \ref{conj:deteq} or to an explicit computation (where Levinson-Durbin could be used for instance).
The time to verify that the tables obtained from the trivial and our accelerated methods coincide is not taken into account; we must check this because at this time we cannot prove the validity of Conjecture \ref{conj:deteq} and/or further enhanced it. We also parallelize both the naive and accelerated algorithms. We notice that our accelerated algorithm on a single core is faster than the trivial algorithm on 160 cores for instance as shown here in the following tables. For $n=2^{14}=16384$, the time to run the trivial algorithm is prohibitive and we did not run the trivial algorithm for length $n=2^{14}$. The meanings of the abbreviations in the following tables are: \textbf{Tri.~S.T.} for trivial algorithm single threaded,
\textbf{Tri.~M.T.} for trivial algorithm multi threaded, \textbf{Acc.~S.T.} for accelerated algorithm single threaded, and \textbf{Acc.~M.T.} for accelerated algorithm multi threaded. Roughly speaking, a thread is a core. All threads share a unique space in memory.
\begin{center}
\begin{longtable}{|c|c|c|c|}
\hline \multicolumn{4}{|c|}{\uppercase{Excerpt of running times for} $n=4096$ (in milliseconds)}\\ \hline
\hline \textbf{Tri.~S.T.} & \textbf{Tri.~M.T.} & \textbf{Acc.~S.T.} & \textbf{Acc.~M.T.}\\ \hline
\endfirsthead
\hline \textbf{Tri.~S.T.} & \textbf{Tri.~M.T.} & \textbf{Acc.~S.T.} & \textbf{Acc.~M.T.}\\ \hline
\endhead
\hline \multicolumn{4}{c}{Continued on next page}
\endfoot
\endlastfoot
\hline
$15931964.434875$ & $793513.102411$ & $209495.875914$ & $79493.525290$\\ \hline
$15930582.584026$ & $793131.331561$ & $212010.715863$ & $79784.405378$\\ \hline
$15931070.671137$ & $793029.664463$ & $211231.649798$ & $79538.533773$\\ \hline\hline
$\approx 4$h $30$min & $\approx14$min& $\approx 4$min & $\approx 1$min $33$sec\\ \hline
\end{longtable}
\begin{longtable}{|p{4cm}|p{4cm}|}
\hline \multicolumn{2}{|c|}{\uppercase{Excerpt of running times for} $n=16384$ (in milliseconds)}\\ \hline
\hline \textbf{Acc.~S.T.} & \textbf{Acc.~M.T.}\\ \hline
\endfirsthead
\hline \textbf{Acc.~S.T.} & \textbf{Acc.~M.T.}\\ \hline
\endhead
\hline \multicolumn{2}{c}{Continued on next page}
\endfoot
\endlastfoot
\hline
$95397251.744779$ & $18823959.088013$\\ \hline
$92977447.400304$ & $18820690.826587$\\ \hline
$93372757.004548$ & $18817978.092718$\\ \hline\hline
$\approx 26$h $30$min & $\approx 5$h $36$min\\ \hline
\end{longtable}
\end{center}

The hardware specification for the computer we used is: Intel(R) Core(TM) i7-8700 CPU @{} 3.20Ghz, 160 cores, 1TB RAM.

We coded Algorithm \ref{alg:dynamicprog} over $\mathbb{F}_{2}$, and a compile switch can be enable to avoid using the library NTL or to use it. Our code is available at \url{https://github.com/63EA13D5/}. For the worst-case instances, we generate the sequences using the following parameters:
\begin{enumerate}
\item[1.] Elements indexed from $0$ to ${7n}/{16}$ inclusively are generated randomly.
\item[2.] Elements indexed from ${7n}/{16} + 1$ to ${9n}/{16}$ inclusively are linearly filled using a non-trivial generating vector of length $d={n}/{8}$. The generating vector is randomly created and the rightmost coordinate is set to the unit element in base field.
\item[3.] Elements indexed from ${9n}/{16}$ to $n$ inclusively are generated randomly.
\end{enumerate}
The ratio of the number of entries in the big square over the number of entries for the table of a given instance is about ${1}/{16}$ up to a few decimals. For each value of $n$, a sample of sequences is used to estimate the running time by evaluating the averages over the sample, one average for the trivial and one average for our method. For comparison, each method is applied to a sequence from the sample. The ratios of the averages of the new method by the trivial are given. We generate a sample of $1000$ linearly filled vectors as described above for each value of $n$. Zero counts are not shown in the tables.


\footnotesize

\begin{center}
\begin{longtable}{|l|r||l|r|}
\caption{Time complexity and distribution of counts\---hard instances.}
\label{tab:res:4096H}\\
\endfirsthead
\endhead
\hline \multicolumn{4}{c}{Continued on next page}
\endfoot
\endlastfoot
\hline
\multicolumn{2}{|l|}{Sample size} & \multicolumn{2}{|l|}{$10140$}\\ \hline
\multicolumn{2}{|l|}{Sequence length} & \multicolumn{2}{|l|}{$4096$} \\ \hline
\multicolumn{2}{|l|}{Generating vector length} & \multicolumn{2}{|l|}{$256$} \\ \hline
\multicolumn{2}{|l|}{Subsequence leftmost index} & \multicolumn{2}{|l|}{$1792$} \\ \hline
\multicolumn{2}{|l|}{Subsequence rightmost index} & \multicolumn{2}{|l|}{$2304$} \\ \hline
\multicolumn{2}{|l|}{Number of entries} & \multicolumn{2}{|l|}{$4192256$} \\ \hline
\hline
\multicolumn{2}{|l|}{Average time for accelerated method (ms)} & \multicolumn{2}{|l|}{$482211.336405$\quad\textrm{(i)}}\\ \hline
\multicolumn{2}{|l|}{Average time for trivial method (ms)} & \multicolumn{2}{|l|}{$39627789.122209$\quad\textrm{(ii)}} \\ \hline
\multicolumn{2}{|l|}{Ratio \textrm{i}/\textrm{ii}}& \multicolumn{2}{|l|}{$0.012169$} \\ \hline
\hline
\multicolumn{2}{|l|}{Average counts NSEW} & \multicolumn{2}{|l|}{$1720642.408481$} \\ \hline
\multicolumn{2}{|l|}{Average counts square filling} & \multicolumn{2}{|l|}{$925814.212032$} \\ \hline
\multicolumn{2}{|l|}{Average counts direct} & \multicolumn{2}{|l|}{$60875.169231$} \\ \hline
\hline
\multicolumn{2}{|l|}{Average counts $2\times 2$ grid} & \multicolumn{2}{|l|}{$467181.981164$} \\ \hline
\multicolumn{2}{|l|}{Average counts $3\times 3$ grid} & \multicolumn{2}{|l|}{$382111.856114$} \\ \hline
\multicolumn{2}{|l|}{Average counts $4\times 4$ grid} & \multicolumn{2}{|l|}{$265305.747436$} \\ \hline
\multicolumn{2}{|l|}{Average counts $5\times 5$ grid} & \multicolumn{2}{|l|}{$169767.046746$} \\ \hline
\multicolumn{2}{|l|}{Average counts $6\times 6$ grid} & \multicolumn{2}{|l|}{$94158.297732$} \\ \hline
\multicolumn{2}{|l|}{Average counts $7\times 7$ grid} & \multicolumn{2}{|l|}{$51134.827416$} \\ \hline
\hline
Average counts $2$-cross & $3772.531657$ & Average counts $3$-cross & $2480.830473$ \\ \hline
Average counts $4$-cross & $1473.981558$ & Average counts $5$-cross & $821.684813$ \\ \hline
Average counts $6$-cross & $4415.502071$ & Average counts $7$-cross & $4449.331460$ \\ \hline
Average counts $8$-cross & $16683.110947$ & Average counts $9$-cross & $9334.250197$ \\ \hline
Average counts $10$-cross & $5165.383136$ & Average counts $11$-cross & $2828.874063$ \\ \hline
Average counts $12$-cross & $1536.936785$ & Average counts $13$-cross & $828.868540$ \\ \hline
Average counts $14$-cross & $444.807101$ & Average counts $15$-cross & $237.420710$ \\ \hline
Average counts $16$-cross & $125.906312$ & Average counts $17$-cross & $67.308679$ \\ \hline
Average counts $18$-cross & $35.672189$ & Average counts $19$-cross & $18.771893$ \\ \hline
Average counts $20$-cross & $9.915779$ & Average counts $21$-cross & $5.115089$ \\ \hline
Average counts $22$-cross & $2.713708$ & Average counts $23$-cross & $1.373570$ \\ \hline
Average counts $24$-cross & $0.696746$ & Average counts $25$-cross & $0.388955$ \\ \hline
Average counts $26$-cross & $0.246154$ & Average counts $27$-cross & $0.098422$ \\ \hline
Average counts $28$-cross & $0.054734$ & Average counts $29$-cross & $0.043590$ \\ \hline
Average counts $30$-cross & $0.020809$ & Average counts $31$-cross & $0.009369$ \\ \hline
Average counts $32$-cross & $0.006312$ & Average counts $256$-cross & $21.695464$ \\ \hline
Average counts $257$-cross & $53.271203$ & Average counts $258$-cross & $70.017061$ \\ \hline
Average counts $259$-cross & $74.697732$ & Average counts $260$-cross & $67.246943$ \\ \hline
Average counts $261$-cross & $56.079487$ & Average counts $262$-cross & $45.345661$ \\ \hline
Average counts $263$-cross & $35.773866$ & Average counts $264$-cross & $26.859369$ \\ \hline
Average counts $265$-cross & $20.313807$ & Average counts $266$-cross & $14.738166$ \\ \hline
Average counts $267$-cross & $9.678205$ & Average counts $268$-cross & $7.546055$ \\ \hline
Average counts $269$-cross & $5.817850$ & Average counts $270$-cross & $3.680079$ \\ \hline
Average counts $271$-cross & $2.569625$ & Average counts $272$-cross & $1.961440$ \\ \hline
Average counts $273$-cross & $1.428895$ & Average counts $274$-cross & $1.082840$ \\ \hline
Average counts $275$-cross & $0.923471$ & Average counts $276$-cross & $0.490335$ \\ \hline
Average counts $277$-cross & $0.383432$ & Average counts $278$-cross & $0.439250$ \\ \hline
Average counts $279$-cross & $0.247929$ & Average counts $280$-cross & $0.082840$ \\ \hline
Average counts $281$-cross & $0.027811$ & Average counts $282$-cross & $0.027811$ \\ \hline
Average counts $285$-cross & $0.028205$ & Average counts $286$-cross & $0.028205$ \\ \hline
Average counts $287$-cross & $0.028402$ & Average counts $288$-cross & $0.028402$ \\ \hline
Average counts $293$-cross & $0.028994$ & Average counts $294$-cross & $0.028994$ \\ \hline
\hline
\multicolumn{2}{|r|}{Sum over all average counts} & \multicolumn{2}{|l|}{$4192256$} \\
\hline
\end{longtable}
\end{center}

\begin{center}
\begin{longtable}{|l|r||l|r|}
\caption{Time complexity and distribution of counts\---easy instances}
\label{tab:res:4096E}\\
\endfirsthead
\endhead
\hline \multicolumn{4}{c}{Continued on next page}
\endfoot
\endlastfoot
\hline
\multicolumn{2}{|l|}{Sample size} & \multicolumn{2}{|l|}{$10140$}\\ \hline
\multicolumn{2}{|l|}{Sequence length} & \multicolumn{2}{|l|}{$4096$} \\ \hline
\multicolumn{2}{|l|}{Generating vector length} & \multicolumn{2}{|l|}{$256$} \\ \hline
\multicolumn{2}{|l|}{Subsequence leftmost index} & \multicolumn{2}{|l|}{$256$} \\ \hline
\multicolumn{2}{|l|}{Subsequence rightmost index} & \multicolumn{2}{|l|}{$4096$} \\ \hline
\multicolumn{2}{|l|}{Number of entries} & \multicolumn{2}{|l|}{$4192256$} \\ \hline
\hline
\multicolumn{2}{|l|}{Average time for accelerated method (ms)} & \multicolumn{2}{|l|}{$4331.896007$\quad\textrm{(i)}}\\ \hline
\multicolumn{2}{|l|}{Average time for trivial method (ms)} & \multicolumn{2}{|l|}{$25765730.530170$\quad\textrm{(ii)}} \\ \hline
\multicolumn{2}{|l|}{Ratio \textrm{i}/\textrm{ii}}& \multicolumn{2}{|l|}{$0.000168$} \\ \hline
\hline
\multicolumn{2}{|l|}{Average counts NSEW} & \multicolumn{2}{|l|}{$407852.065385$} \\ \hline
\multicolumn{2}{|l|}{Average counts square filling} & \multicolumn{2}{|l|}{$3419936.066568$} \\ \hline
\multicolumn{2}{|l|}{Average counts direct} & \multicolumn{2}{|l|}{$19444.707988$} \\ \hline
\hline
\multicolumn{2}{|l|}{Average counts $2\times 2$ grid} & \multicolumn{2}{|l|}{$106227.368146$} \\ \hline
\multicolumn{2}{|l|}{Average counts $3\times 3$ grid} & \multicolumn{2}{|l|}{$86949.429093$} \\ \hline
\multicolumn{2}{|l|}{Average counts $4\times 4$ grid} & \multicolumn{2}{|l|}{$60390.460256$} \\ \hline
\multicolumn{2}{|l|}{Average counts $5\times 5$ grid} & \multicolumn{2}{|l|}{$38652.743195$} \\ \hline
\multicolumn{2}{|l|}{Average counts $6\times 6$ grid} & \multicolumn{2}{|l|}{$21438.953156$} \\ \hline
\multicolumn{2}{|l|}{Average counts $7\times 7$ grid} & \multicolumn{2}{|l|}{$11646.030572$} \\ \hline
\hline
Average counts $2$-cross & $3772.586193$ & Average counts $3$-cross & $2479.999310$ \\ \hline
Average counts $4$-cross & $1473.537081$ & Average counts $5$-cross & $821.756805$ \\ \hline
Average counts $6$-cross & $1341.515483$ & Average counts $7$-cross & $1185.020414$ \\ \hline
Average counts $8$-cross & $3878.167554$ & Average counts $9$-cross & $2161.438363$ \\ \hline
Average counts $10$-cross & $1191.903748$ & Average counts $11$-cross & $651.711736$ \\ \hline
Average counts $12$-cross & $351.648225$ & Average counts $13$-cross & $189.478107$ \\ \hline
Average counts $14$-cross & $102.336391$ & Average counts $15$-cross & $54.392998$ \\ \hline
Average counts $16$-cross & $28.796746$ & Average counts $17$-cross & $15.180178$ \\ \hline
Average counts $18$-cross & $8.078008$ & Average counts $19$-cross & $4.166469$ \\ \hline
Average counts $20$-cross & $2.130868$ & Average counts $21$-cross & $1.140039$ \\ \hline
Average counts $22$-cross & $0.562821$ & Average counts $23$-cross & $0.254734$ \\ \hline
Average counts $24$-cross & $0.159073$ & Average counts $25$-cross & $0.104832$ \\ \hline
Average counts $26$-cross & $0.051972$ & Average counts $27$-cross & $0.035207$ \\ \hline
Average counts $28$-cross & $0.019428$ & Average counts $29$-cross & $0.002860$ \\ \hline
Average counts $1792$-cross & $0.645759$ & Average counts $1793$-cross & $0.501775$ \\ \hline
Average counts $1794$-cross & $0.350099$ & Average counts $1795$-cross & $0.212623$ \\ \hline
Average counts $1796$-cross & $0.121893$ & Average counts $1797$-cross & $0.073373$  \\ \hline
Average counts $1798$-cross & $0.039645$ & Average counts $1799$-cross & $0.022091$ \\ \hline
Average counts $1800$-cross & $0.015385$ & Average counts $1801$-cross & $0.007298$ \\ \hline
Average counts $1802$-cross & $0.005128$ & Average counts $1803$-cross & $0.001578$ \\ \hline
Average counts $1804$-cross & $0.002564$ & Average counts $1805$-cross & $0.000394$ \\ \hline
Average counts $1806$-cross & $0.000197$ & Average counts $1812$-cross & $0.000197$ \\ \hline
\hline
\multicolumn{2}{|r|}{Sum over all average counts} & \multicolumn{2}{|l|}{$4192256$} \\
\hline
\end{longtable}
\end{center}

\normalsize

We observe that we are about $83$ times faster on typical hard instances and about $5947$ times faster on easy ones. In practice, to detect linearity or to solve backward for the generating vector, we only need to stop at the first level that contains a long run of zeros.

\end{document}